\documentclass[11pt]{article}

\usepackage{graphicx} 
\usepackage[T1]{fontenc}
\usepackage[utf8x]{inputenc}
\usepackage{amsmath, amssymb, amsthm}
\usepackage{csquotes}
\usepackage[english]{babel}
\usepackage{eurosym}
\usepackage{fullpage}
\usepackage{palatino}
\usepackage{tikz}
\usepackage{verbatim}
\usepackage{fancybox,soul,color}

\usepackage{subfig}
\usepackage{enumerate}
\usepackage{epsfig}
\usepackage[active]{srcltx}
\usepackage{amsfonts}
\usepackage{amsmath}
\usepackage[mathlines]{lineno}
\usepackage{xspace}
\usepackage{tikz}

\usetikzlibrary{arrows,shapes,snakes,automata,backgrounds,petri}

\DeclareMathOperator{\mad}{mad}
\DeclareMathOperator{\ad}{ad}

\DeclareMathOperator{\D}{\Delta(G)}

\newtheorem{lemma}{Lemma}

\newtheorem{theorem}{Theorem}
\newtheorem{claim}{Claim}

\newtheorem{conjecture}{Conjecture}
\newtheorem{question}{Question}

\def\claimb{$$\vcenter\bgroup\advance\hsize by -8em\noindent
\refstepcounter{claimb}\ignorespaces\it}        
\makeatletter
\def\endclaimb{\rm\egroup\leqno(\theclaim)$$\global\@ignoretrue}
\makeatother

    {\noindent \emph{Proof.} {}{#1}{}}{\hfill
    $\Diamond$\vspace{1em}}

\begin{document}

\title{List coloring the square of sparse graphs with large degree\thanks{This work was partially supported by the ANR grant EGOS 12 JS02 002 01}}

\author{Marthe Bonamy, Benjamin Lévêque, Alexandre Pinlou\thanks{Second affiliation: D\'epartement MIAp, Universit\'e Paul-Val\'ery, Montpellier 3}\\ \normalsize{LIRMM, Universit\'e Montpellier 2, CNRS}\\ \small{\{marthe.bonamy, benjamin.leveque, alexandre.pinlou\}@lirmm.fr}}

\maketitle

\begin{abstract}
We consider the problem of coloring the squares of graphs of bounded maximum average degree, that is, the problem of coloring the vertices while ensuring that two vertices that are adjacent or have a common neighbour receive different colors. 

Borodin et al. proved in 2004 and 2008 that the squares of planar graphs of girth at least seven and sufficiently large maximum degree $\Delta$ are list $(\Delta+1)$-colorable, while the squares of some planar graphs of girth six and arbitrarily large maximum degree are not. By Euler's Formula, planar graphs of girth at least $6$ are of maximum average degree less than $3$, and planar graphs of girth at least $7$ are of maximum average degree less than $\frac{14}{5}<3$.

We strengthen their result and prove that there exists a function $f$ such that the square of any graph with maximum average degree $m<3$ and maximum degree $\Delta\geq f(m)$ is list $(\Delta+1)$-colorable. This bound of $3$ is optimal in the sense that the above-mentioned planar graphs with girth $6$ have maximum average degree less than $3$ and arbitrarily large maximum degree, while their square cannot be $(\Delta+1)$-colored. The same holds for list injective $\Delta$-coloring.
\end{abstract} 

\section{Introduction}

The square of a graph $G$ is defined as a graph with the same set of vertices as $G$, where two vertices are adjacent if and only if they are adjacent or have a common neighbor in $G$. A \emph{k-coloring} of the square of a graph $G$ (also known as $2$-distance $k$-coloring of $G$) is therefore a coloring of the vertices of $G$ with $k$ colors such that two vertices that are adjacent or have a common neighbor receive distinct colors. We define $\chi^2(G)$ as the smallest $k$ such that the square of $G$ admits a $k$-coloring. For example, the square of a cycle of length $5$ cannot be colored with less than $5$ colors as any two vertices are either adjacent or have a common neighbor: its square is the clique of size $5$.

The study of $\chi^2(G)$ on planar graphs was initiated by Wegner in 1977~\cite{w77}, and has been actively studied since. The \emph{maximum degree} of a graph $G$ is denoted $\D$. Note that any graph $G$ satisfies $\chi^2(G) \geq \D+1$. Indeed, if we consider a vertex of maximal degree and its neighbors, they form a set of $\D+1$ vertices, any two of which are adjacent or have a common neighbor. Hence at least $\D+1$ colors are needed to color the square of $G$. It is therefore natural to ask when this lower bound is reached. For that purpose, we can study, as suggested by Wang and Lih~\cite{wl03}, what conditions on the sparseness of the graph can be sufficient to ensure the equality holds. The sparseness of a planar graph can for example be measured by its girth. The \emph{girth} of a graph $G$, denoted $g(G)$, is the length of a shortest cycle.

\begin{conjecture}[Wang and Lih~\cite{wl03}]\label{conj:wl03}
For any integer $k \geq 5$, there exists an integer $M(k)$ such that for every planar graph $G$ satisfying $g(G) \geq k$ and $\D \geq M(k)$, $\chi^2(G)=\D+1$.
\end{conjecture}

Conjecture~\ref{conj:wl03} was proved in \cite{bgint04,bin04,dkns08} to be true for $k \geq 7$ and false for $k=6$. An extension of the $k$-coloring of the square is the \emph{list $k$-coloring of the square}, where instead of having the same set of $k$ colors for the whole graph, every vertex is assigned some set of $k$ colors and has to be colored from it. Given a graph $G$, we call $\chi^2_\ell(G)$ the minimal integer $k$ such that the square of $G$ admits a list $k$-coloring for any list assignment. Obviously,  coloring is a subcase of list coloring (where the same color list is assigned to every vertex), so for any graph $G$, we have $\chi^2_\ell(G) \geq \chi^2(G)$. Thus, in the case of list-coloring, Conjecture~\ref{conj:wl03} is also false for $k=6$, and Borodin, Ivanova and Neustroeva~\cite{bin08} proved it to be true for $k \geq 7$.

Another way to measure the sparseness of a graph is through its maximum average degree as defined below. The average degree of a graph $G$, denoted $\ad(G)$, is $\frac{\sum_{v \in V}d(v)}{|V|}=\frac{2|E|}{|V|}$. The \emph{maximum average degree} of a graph $G$, denoted $\mad(G)$, is the maximum of $\ad(H)$ on every subgraph $H$ of $G$. Euler's formula links girth and maximum average degree in the case of planar graphs.

\begin{lemma}[Folklore]\label{lem:euler}
For every planar graph $G$, $(\mad(G)-2)(g(G)-2)<4$.
\end{lemma}

The question raised by Conjecture~\ref{conj:wl03} and now solved could be reworded as follows: what is the minimum $k$ such that any graph $G$ with $g(G) \geq k$ and large enough $\Delta(G)$ (depending only on $g(G)$) satisfies
$\chi_\ell^2=\D+1$? A consequence of Lemma~\ref{lem:euler} is that we can transpose any theorem holding for an upper bound on $\mad(G)$ into a theorem holding for planar graphs with lower-bounded girth. It is then natural to transpose the question to the maximum average degree, as it is a more refined measure of sparseness. More precisely, what is the supremum $M$ such that any graph $G$ with $\mad(G)<M$ and large enough $\Delta(G)$ (depending only on $\mad(G)$) satisfies $\chi_\ell^2=\D+1$?

The authors \cite{blp13} proved that $\frac{14}{5} \leq M$, which was recently also proved by Cranston and \v{S}krekovski~\cite{cs13}. We know that $M \leq 3$ due to the family of graphs that appears in~\cite{bgint04} (see Figure~\ref{fig:mad3}), which are of maximum average degree $<3$, of increasing maximum degree, and whose squares are not $(\Delta+1)$-colorable. We prove here that $3 \leq M$, thus obtaining the exact value of $M$, which is $3$.\\

\begin{figure}[!h]
\center
\begin{tikzpicture}
    \tikzstyle{whitenode}=[draw,circle,fill=white,minimum size=9pt,inner sep=0pt]
    \tikzstyle{blacknode}=[draw,circle,fill=black,minimum size=6pt,inner sep=0pt]
 \draw (0,0) node[blacknode] (a) {}
-- ++(0:1.5cm) node[blacknode] (b) {}
-- ++(0:1.5cm) node[blacknode] (c) {};

\draw (a) 
-- ++(30:1.1cm) node[blacknode] (d) {};

\draw (a) 
-- ++(60:1.7cm) node[blacknode] (f) {};

\draw (c) 
-- ++(150:1.1cm) node[blacknode] (e) {};

\draw (c) 
-- ++(120:1.7cm) node[blacknode] (g) {};

\draw (d) edge  node {} (e);

\draw (f) edge  node {} (g);

\draw[loosely dotted][thick] (1.5,0.8) edge  node {} (1.5,1.4);

\draw (3.2,0) node[right=1pt] (g) {\footnotesize $1$};
\draw (3.2,0.6) node[right=1pt] (g) {\footnotesize $2$};
\draw (3.2,1.4722) node[right=1pt] (g) {\footnotesize $p$};
\draw[loosely dotted][thick] (3.4,0.8) edge  node {} (3.4,1.3);
\end{tikzpicture}
\caption{A graph $G_p$ with $\Delta(G_p) = p$, $\mad(G_p)=3-\frac{5}{2p+1}$ and $\chi^2(G_p)=\Delta(G_p)+2$.}
\label{fig:mad3}
\end{figure}
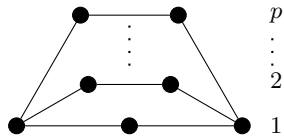

\begin{theorem}\label{thm:m145}
There exists a function $f$ such that for any $\epsilon>0$, every graph $G$ with $mad(G)<3-\epsilon$ and $\Delta(G) \geq f(\epsilon)$ satisfies $\chi^2_\ell(G)= \Delta(G)+1$.
\end{theorem}

This answers the transposition of Conjecture~\ref{conj:wl03} to graphs with an upper-bounded maximum average degree. As the maximum average degree is not discrete, we obtain a sharper value than for planar graphs of bounded girth.

More generally, is it possible to get similar results when allowing an additional constant number of colors, as was done by Wang and Lih in~\cite{wl03} for planar graphs? More precisely, what is, for any $C \geq 2$, the supremum $M(C)$ such that any graph $G$ with $\mad(G) < M(C)$ and sufficiently large $\Delta(G)$ (depending only on $mad(G)$) satisfies $\chi^2_\ell(G) \leq \D+C$?

%\sethlcolor{green}
The authors proved in \cite{blp13} that $\lim_{C\to\infty} M(C)= 4$. Interestingly, while graphs with $\mad(G)<4-\epsilon$ satisfy $\chi^2_\ell(G)\leq \Delta(G)+\mathcal{O}(\frac{1}{\epsilon})$, some graphs with $\mad(G)<4$ and arbitrarily large maximum degree have $\chi^2(G)\geq \frac{3\Delta(G)}{2}$. This is true even with a restriction to planar graphs with girth at least $4$.

Charpentier~\cite{c12} generalized the family of graphs presented in Figure~\ref{fig:mad3} to obtain for each $C$ a family of graphs which are of maximum average degree less than $\frac{4C+2}{C+1}$, of increasing maximum degree, and whose square requires $\Delta+C+1$ colors to be colored (see Figure~\ref{fig:mad4i}). Consequently, for every $C$, we have $M(C) \leq \frac{4C+2}{C+1}$.

\begin{figure}[!h]
\center
\begin{tikzpicture}[scale=2.5,rotate=180]
\tikzstyle{blacknode}=[draw,circle,fill=black,minimum size=6pt,inner sep=0pt]
\draw (0,0) node[blacknode] (u) [label=right:$u$] {}
-- ++(-45:1cm) node[blacknode] (v1) [label=90:$v_1$] {}
-- ++(0:1cm) node[blacknode] (v1w1) {}
-- ++(0:1cm) node[blacknode] (w1) [label=90:$w_1$] {}
-- ++(45:1cm) node[blacknode] (x) [label=left:$x$] {};
 
%\draw (u) edge [bend right=120] node {} (x);

\draw (u) 
-- ++(45:1cm) node[blacknode] (vp) [label=-90:$v_p$] {}
-- ++(0:1cm) node[blacknode] (vpwC) {}
-- ++(0:1cm) node[blacknode] (wC) [label=-90:$w_C$] {};
\draw (wC) edge  node {} (x);

\draw (u) 
-- ++(-20:0.75cm) node[blacknode] (v2) [label=90:$v_2$] {}
-- ++(0:1cm) node[blacknode] (v2w2) {}
-- ++(0:1cm) node[blacknode] (w2) [label=90:$w_2$] {};
\draw (w2) edge  node {} (x);

\draw (vp)
-- ++(-38:0.7cm) node[blacknode] (vpw1) {};
\draw (vpw1) edge  node {} (w1);
\draw (vp)
-- ++(-28:0.7cm) node[blacknode] (vpw2) {};
\draw (vpw2) edge  node {} (w2);
\draw (wC)
-- ++(180+28:0.7cm) node[blacknode] (v2wC) {};
\draw (v2wC) edge  node {} (v2);

\draw (v2)
-- ++(-12:0.7cm) node[blacknode] (v2w1) {};
\draw (v2w1) edge  node {} (w1);
\draw (v1)
-- ++(12:1.4cm) node[blacknode] (v1w2) {};
\draw (v1w2) edge  node {} (w2);

\draw (v1)
-- ++(35:1.6cm) node[blacknode] (v1wC) {};
\draw (v1wC) edge  node {} (wC);
 
\draw[loosely dotted][very thick] (0.7,-0.1) edge  node {} (0.7,0.5); 
\draw[loosely dotted][very thick] (2.7,-0.1) edge  node {} (2.7,0.5); 

\draw plot [smooth, tension=0.5] coordinates {(x) (2.7,1) (0.7,1) (u)};
%\draw (x) parabola bend (1.7,1.5) (u) node[below right] {};
\end{tikzpicture}
\caption{For $p \geq C$, a graph $G_{p,C}$ with $\Delta(G_{p,C}) = p$, $\mad(G_{p,C})=\frac{
(2C+1)(2p+1)+1
%4Cp+2p+2C+2
}{
(C+1)(p+1)+1
%Cp+p+C+2
}$ and $\chi^2(G_{p,C})=p+C+1$.% The subgraph induced by $\{v_1,v_2,\ldots,v_p,w_1,w_2,\ldots,w_C\}$ is a subdivision of a complete bipartite graph.
}
\label{fig:mad4i}
\end{figure}
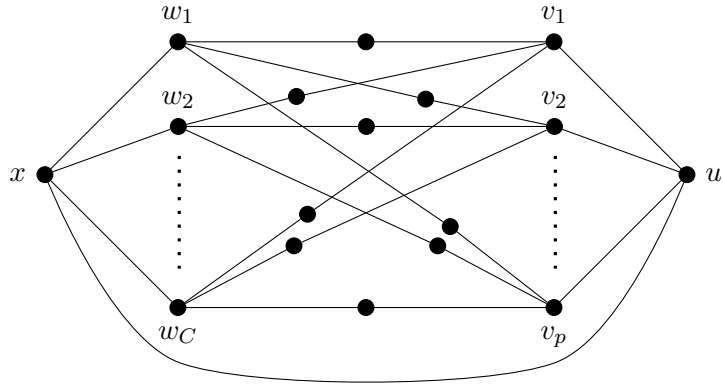

This result, and the fact that $\frac{4C+2}{C+1}$ equals $M(C)$ when $C=1$ and when $C$ tends to infinite, raise the following question.

\begin{question}
Is it true that $M(C)=\frac{4C+2}{C+1}$ for any $C\geq 1$?
\end{question}

Theorem~\ref{thm:m145} is proved using a discharging method. The discharging method was introduced in the beginning of the 20$^{th}$ century. It is notably known for being used to prove the Four Color Theorem (\cite{ah77} and \cite{ahk77}). When the discharging rules are local (ie the weight cannot travel arbitrarily far), as is most commonly, we say the discharging method is \emph{local}. Borodin, Ivanova and Kostochka introduced in~\cite{bik05} the notion of \emph{global} discharging, which is when there is no bound on the size of the discharging rules (ie the weight can travel arbitrarily far along the graph). When it is mixed, ie the discharging rules are of bounded size but take into account structures of unbounded size in the graph, we say the discharging method is \emph{semi-global} (see~\cite{b89} for a first occurrence of such a proof). Our proof of Theorem~\ref{thm:m145} is presented in Section~\ref{sect:thm} as global for simplicity, but could actually be made semi-global by more careful discharging. The global discharging argument is of the same vein as a nice proof of Borodin, Kostochka and Woodall~\cite{bkw97} later simplified by Woodall~\cite{w10}. We explain in Section~\ref{sect:inj} how this proof can be transposed to injective colorings.

\section{Proof of Theorem~\ref{thm:m145}}\label{sect:thm}

We prove that there exists a function $f$ such that for any $\epsilon>0$, every graph $G$ with $mad(G)<3-\epsilon$ and $\Delta(G) \geq f(\epsilon)$ satisfies $\chi^2_\ell(G)= \Delta(G)+1$. In the following, we try to simplify the proof rather than improve the function $f$.

For technical reasons, we will have to consider $\epsilon \leq \frac{1}{20}$. For $\epsilon>\frac{1}{20}$, it suffices to set $f(\epsilon)=f(\frac{1}{20})$. Indeed, if $\epsilon>\frac{1}{20}$, then for every graph with $mad(G)<3-\epsilon$ and $\Delta(G) \geq f(\epsilon)$, we have in particular $mad(G)<3-\frac{1}{20}$ and $\Delta(G)\geq f(\frac{1}{20})$, thus the conclusion holds. From now on, we consider $\epsilon \leq \frac{1}{20}$.

Let $f : \epsilon \mapsto \frac{3}{\epsilon^2}$. Assume by contradiction that there exist a constant $\frac{1}{20} \geq \epsilon > 0$ and a graph $\Gamma$ with $mad(\Gamma)<3-\epsilon$ and $\Delta(\Gamma) \geq f(\epsilon)$ that satisfies $\chi^2_\ell(\Gamma)>\Delta(\Gamma)+1$. There is a \emph{minimal} subgraph $G$ of $\Gamma$ such that $\chi^2_\ell(G)>\Delta(\Gamma)+1$, in the sense that the square of every proper subgraph of $G$ is list $(\Delta(\Gamma)+1)$-colorable. For $k=\Delta(\Gamma)$, the graph $G$ satisfies $\Delta(G)\leq k$ and $\chi^2_\ell(G)>k+1$, while the square of all its proper subgraphs are list $(k+1)$-colorable. We aim at proving that $mad(G) \geq 3-\epsilon$, a contradiction to the fact that $G$ is a subgraph of $\Gamma$ with $mad(\Gamma)<3-\epsilon$.

Let $M=\frac{6}{\epsilon}$. Note that since $\epsilon \leq \frac{1}{20}$, we have  $k=\Delta(\Gamma)\geq f(\epsilon)=\frac{3}{\epsilon^2} \geq \frac{18}{\epsilon}=3\times M$.

In Subsection~\ref{sect:def}, we introduce the terminology and notation. In Subsection~\ref{sect:dis}, we use the structural observations from Subsections~\ref{sect:conf} and \ref{sect:glob} to derive with a discharging argument that such a graph has maximum average degree at least $3-\epsilon$, which concludes the proof.
\subsection{Terminology and notations}\label{sect:def}

In the figures, we draw in black a vertex that has no other neighbor than the ones already represented, in white a vertex that might have other neighbors than the ones represented. When there is a label inside a white vertex, it is an indication on the number of neighbors it has. The label '$i$' means "exactly $i$ neighbors", the label '$i^+$' (resp. '$i^-$') means that it has at least (resp. at most) $i$ neighbors. Note that the white vertices may coincide with other vertices.

A \emph{constraint} of a vertex $u$ is an already colored vertex that is adjacent to or has a common neighbor with $u$. Two constraints with the same color count as one.

Given a vertex $u$, the \emph{neighborhood} $N(u)$ is the set of vertices that are adjacent to $u$. For $p\geq 1$, a \emph{$p$-link} $x-a_1-\ldots-a_p-y$ between $x$ and $y$ is a path such that $d(a_1)=\ldots=d(a_p)=2$. When a $p$-link exists between two vertices $x$ and $y$, we say they are \emph{$p$-linked}.

\subsection{Forbidden Configurations}\label{sect:conf}

We define configurations \textbf{($C_1$)} to \textbf{($C_3$)} (see Figure~\ref{fig:config}). 
\begin{itemize}
\item \textbf{($C_1$)} is a vertex $u$ of degree $0$ or $1$. 
\item \textbf{($C_2$)} is a vertex $w_1$ of degree at most $k-1$ that is $2$-linked (through $w_1$-$u_1$-$u_2$-$w_2$) to a vertex $w_2$ of degree at most $k-2$.
\item \textbf{($C_3$)} is a vertex $u$ with $3 \leq d(u) \leq M$ that is $1$-linked (through $u$-$v_i$-$w_i$) to $(d(u)-2)$ vertices $(w_i)_{1 \leq i \leq d(u)-2}$ of degree at most $M$, and such that the sum of the degrees of its two other neighbors $x$ and $y$ is at most $k-M+2$.
\end{itemize}

\captionsetup[subfloat]{labelformat=empty}
\begin{figure}[!h]
\centering
\subfloat[][\textbf{($C_1$)}]{
\centering
\begin{tikzpicture}[scale=0.95]
\tikzstyle{whitenode}=[draw,circle,fill=white,minimum size=8pt,inner sep=0pt]
\tikzstyle{blacknode}=[draw,circle,fill=black,minimum size=6pt,inner sep=0pt]
\tikzstyle{tnode}=[draw,ellipse,fill=white,minimum size=8pt,inner sep=0pt]
\tikzstyle{texte} =[fill=white, text=black]
   \draw (0,0) node[whitenode] (u) [label=right:$u$] {\small{$1^-$}};
\end{tikzpicture}
%\caption{\textbf{($C_1$)}}
\label{fig:cc1}
}
\qquad
\subfloat[][\textbf{($C_2$)}]{
\centering
\begin{tikzpicture}[scale=0.95]
\tikzstyle{whitenode}=[draw,circle,fill=white,minimum size=8pt,inner sep=0pt]
\tikzstyle{blacknode}=[draw,circle,fill=black,minimum size=6pt,inner sep=0pt]
\tikzstyle{tnode}=[draw,ellipse,fill=white,minimum size=8pt,inner sep=0pt]
\tikzstyle{texte} =[fill=white, text=black] 
     \draw (2,-3) node[tnode] (w1) [label=right:$w_1$] {\small{$(k-1)^-$}}
        -- ++(90:1cm) node[blacknode] (u1) [label=right:$u_1$] {}
        -- ++(90:1cm) node[blacknode] (u2) [label=right:$u_2$] {}
        -- ++(90:1cm) node[tnode] (w2) [label=90:$w_2$] {\small{$(k-2)^-$}};
\end{tikzpicture}
%\caption{\textbf{($C_2$)}}
\label{fig:cc2}
}
\qquad
\subfloat[][\textbf{($C_3$)}]{
\centering
\begin{tikzpicture}[scale=0.95]
\tikzstyle{whitenode}=[draw,circle,fill=white,minimum size=8pt,inner sep=0pt]
\tikzstyle{blacknode}=[draw,circle,fill=black,minimum size=6pt,inner sep=0pt]
\tikzstyle{tnode}=[draw,ellipse,fill=white,minimum size=8pt,inner sep=0pt]
\tikzstyle{texte} =[fill=white, text=black]

\draw (0,0) node[blacknode] (u) [label=90:$u$] {}
-- ++(160:1cm) node[tnode] (v3) [label=90:$x$] {};
\draw (u)
-- ++(200:1cm) node[tnode] (v4) [label=-90:$y$] {};

\draw (u)
-- ++(40:1cm) node[blacknode] (v1) [label=90:$v_1$] {}
-- ++(0:1cm) node[tnode] (w1) [label=80:$w_1$] {\small{$M^-$}};

\draw (u)
-- ++(-40:1cm) node[blacknode] (vd) [label=-90:$v_{d(u)-2}$] {}
-- ++(0:1cm) node[tnode] (wd) [label=-80:$w_{d(u)-2}$] {\small{$M^-$}};

\draw[thick, dotted, bend left] (20:1.3cm) edge node {} (-20:1.3cm);

\node[texte, right=1pt] at (-1.2,-1.7) {$d(x)+d(y) \leq k-M+2$};
\node[texte, right=1pt] at (-0.2,-2.3) {$3 \leq d(u)\leq M$};
\end{tikzpicture}
%\caption{\textbf{($C_3$)}}
\label{fig:cc3}
}
\caption{Forbidden configurations for Theorem~\ref{thm:m145}.}
\label{fig:config}
\end{figure}
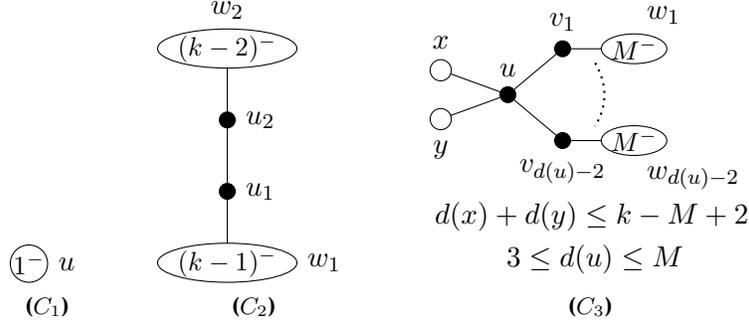
\captionsetup[subfloat]{labelformat=parens}

\begin{lemma}\label{lem:config}
Graph $G$ cannot contain any of Configurations \textbf{($C_1$)} to \textbf{($C_3$)}.
\end{lemma}

\begin{proof}
We assume $G$ contains a configuration, apply the minimality to color a subgraph of $G$, and prove this coloring can be extended to the whole graph, a contradiction.

\begin{claim}
$G$ cannot contain \textbf{($C_1$)}.
\end{claim}
\begin{proof}
Using the minimality of $G$, we color $G \setminus \{u\}$. Since $\Delta(G) \leq k$, and $d(u) \leq 1$, vertex $u$ has at most $k$ constraints. There are $k+1$ colors, so the coloring of $G \setminus \{u\}$ can be extended to $G$.
\end{proof}

\begin{claim}
$G$ cannot contain \textbf{($C_2$)}.
\end{claim}
\begin{proof}
Using the minimality of $G$, we color $G \setminus \{u_1,u_2\}$. Vertex $u_1$ has at most $|\{w_2\}|+d(w_1)\leq 1+(k-1)\leq k$ constraints. Hence we can color $u_1$. Then $u_2$ has at most $|\{w_1,u_1\}|+d(w_2)\leq 2+(k-2)\leq k$ constraints, so we can extend the coloring of $G \setminus \{u_1,u_2\}$ to $G$.
\end{proof}

\begin{claim}
$G$ cannot contain \textbf{($C_3$)}.
\end{claim}
\begin{proof}
Using the minimality of $G$, we color $G \setminus \{v_1,\cdots,v_{d(u)-2}\}$. We did not delete $u$ in order to obtain a coloring where $x$ and $y$ receive different colors, but $u$ might have the same color as some $w_i$, so it needs to be recolored. Vertex $u$ has at most $M-2 +d(x)+d(y) \leq k$ constraints, hence we can recolor $u$. Then every $v_i$ has at most $M+M \leq k$ constraints, so we can extend the coloring of $G \setminus \{v_1,\cdots,v_{d(u)-2}\}$ to $G$. 
\end{proof}

This concludes the proof of Lemma~\ref{lem:config}.
\end{proof}

\subsection{Global structure}\label{sect:glob}

We define three sets $V_1$, $V_2$ and $T$ that will outline some global structure on $G$.
We build inductively the set $V_1$ as follows.

Any vertex $u$ of degree at most $M-1$ belongs to $V_1$ if it is adjacent to $d(u)-1$ vertices $v_1,\ldots,v_{d(u)-1}$ of degree $2$ such that the other neighbors of $v_2,\ldots,v_{d(u)-1}$ belong to $V_1$ while the other neighbor of $v_1$ is of degree at most $M-1$ (but does not necessarily belong to $V_1$).

Let $V_2$ be the subset of $V_1$ defined as the set of vertices $u$ of degree at most $M-1$ that are adjacent to $d(u)-1$ vertices of degree $2$ whose other neighbors all belong to $V_1$. We define $T$ as the set of vertices of degree $2$ whose both neighbors are in $V_1$. See Figure~\ref{fig:V1} for examples of vertices in $V_1$, $V_2$ or $T$. In the figures, we denote by a label $V_1$ (resp. $V_2$, $T$) the fact that a vertex belongs to $V_1$ (resp. $V_2$, $T$). Similarly, we denote by a label $\neg V_1$ a vertex that does not belong to $V_1$. Since $V_2\subset V_1$, we omit the label $V_1$ on vertices labelled $V_2$.

\begin{figure}
\centering
\begin{tikzpicture}[scale=0.95]
\tikzstyle{blacknode}=[draw,circle,fill=black,minimum size=6pt,inner sep=0pt]
\tikzstyle{tnode}=[draw,ellipse,fill=white,minimum size=8pt,inner sep=0pt]
\tikzstyle{texte} =[fill=white, text=black] 
     \draw (0,0) node[tnode] (v1) [label=right:] {\small{$(M-1)^-$}}
        -- ++(90:1cm) node[blacknode] (u1) [label=right:] {}
        -- ++(90:1cm) node[blacknode] (x) [label=right:$V_1$] {}
        -- ++(90:1cm) node[tnode] (w2) [label=90:] {};

     \draw (2,-0.5) node[tnode] (v1) [label=right:] {}
        -- ++(90:1cm) node[blacknode] (u1) [label=right:$V_2$] {}
        -- ++(90:1cm) node[blacknode] (y) [label=right:$T$] {}
        -- ++(90:1cm) node[blacknode] (x) [label=right:$V_2$] {}
        -- ++(90:1cm) node[tnode] (w2) [label=90:] {};
\draw[dotted] (2,1.5) ellipse (.7cm and 1.5cm);

     \draw (4,-0.5) node[tnode] (v1) [label=right:] {}
        -- ++(60:1cm) node[blacknode] (u1) [label=left:$V_2$] {}
        -- ++(60:1cm) node[blacknode] (y) [label=left:$T$] {}
        -- ++(60:1cm) node[blacknode] (x) [label=right:$V_2$] {}
        -- ++(90:1cm) node[tnode] (w2) [label=90:] {};
\draw (x)
        -- ++(-60:1cm) node[blacknode] (y3) [label=right:$T$] {}
        -- ++(-60:1cm) node[blacknode] (x3) [label=right:$V_2$] {}
        -- ++(-60:1cm) node[tnode] (w3) [label=90:] {};
\draw[dotted] (5.5,1) ellipse (2cm and 1.5cm);

     \draw (8,-0.5) node[tnode] (v1) [label=right:] {}
        -- ++(60:1cm) node[blacknode] (u1) [label=left:$V_1$] {}
        -- ++(60:1cm) node[blacknode] (y) [label=left:$\neg T$] {}
        -- ++(60:1cm) node[blacknode] (x) [label=right:$\neg V_1$] {}
        -- ++(135:1cm) node[tnode] (w2) [label=90:] {$3^+$};

\draw (u1)
        -- ++(-60:1cm) node[blacknode] (y) [label=right:$T$] {}
        -- ++(-60:1cm) node[blacknode] (x3) [label=right:$V_2$] {}
        -- ++(-60:1cm) node[tnode] (w3) [label=90:] {};

\draw (x)
        -- ++(-60:1cm) node[blacknode] (y) [label=right:$\neg T$] {}
        -- ++(-60:1cm) node[blacknode] (x3) [label=right:$V_1$] {}
        -- ++(-60:1cm) node[tnode] (w3) [label=90:] {};
\draw(x)
  -- ++(45:1cm) node[tnode] (w3) [label=90:] {$3^+$};
\end{tikzpicture}
\caption{Examples of vertices in $V_1$, $V_2$ or $T$.}
\label{fig:V1}
\end{figure}
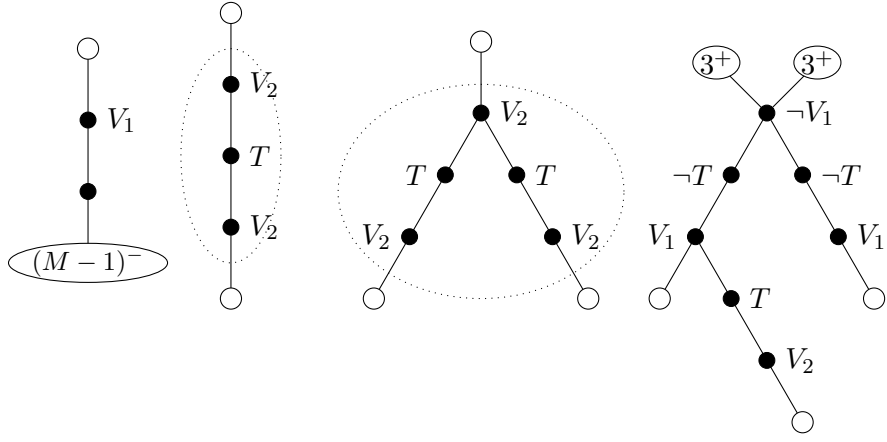

\begin{lemma}\label{lem:V1}
The vertices of $V_1$ satisfy the following:
\begin{itemize}
\item Every vertex of $V_1$ has exactly one neighbor of degree at least $k-M$.
\item The set $V_1$ is a stable set.
\item The sets $V_1$ and $T$ are disjoint.
\end{itemize}
\end{lemma}
\begin{proof}
Assume by contradiction that a vertex $u$ of $V_1$ has no neighbor of degree at least $k-M$. Then $u$ is adjacent to $d(u)-1$ vertices $v_1,\ldots,v_{d(u)-1}$ of degree $2$ whose other neighbors are of degree at most $M-1$, and to another vertex $w$ of degree at most $k-M-1$. We consider two cases depending on whether $d(u)=2$. If $d(u)=2$, then the other neighbor of $v_1$ is a vertex of degree at most $M-1 \leq k-1$ that is $2$-linked to $w$, which is a vertex of degree at most $k-M-1\leq k-2$. By Lemma~\ref{lem:config}, Configuration \textbf{($C_2$)} is not contained in $G$, a contradiction. If $d(u)\geq 3$, then $u$ is a vertex with $3 \leq d(u) \leq M-1\leq M$ that is $1$-linked (through $v_i$, for $1\leq i \leq d(u)-2$) to $d(u)-2$ vertices of degree at most $M-1 \leq M$, and such that the sum of the degrees of its two other neighbors $w$ and $v_{d(u)-1}$ is at most $k-M-1+2 \leq k-M+2$. By Lemma~\ref{lem:config}, Configuration \textbf{($C_3$)} is not contained in $G$, a contradiction. Therefore every vertex $u$ of $V_1$ has a neighbor of degree at least $k-M$. By definition of $V_1$, all the other neighbors of $u$ are of degree $2$. Thus $u$ has a unique neighbor of degree at most $k-M$.
%By Lemma~\ref{lem:config}, and more precisely because Configurations \textbf{($C_2$)} and \textbf{($C_3$)} are not contained in $G$, every vertex of $V_1$ has a neighbor of degree at least $k-M$. 

Since $k \geq 2 \times M$ then $k-M>M-1$ and vertex $u$ has no neighbor $v$ of degree $3 \leq d(v)\leq M-1$. Consequently, two vertices $u,v$ of $V_1$ that are adjacent must both be of degree $2$. By definition of $V_1$, the other neighbors of $u$ and $v$ must be of degree at most $M-1$, a contradiction. It follows that $V_1$ is a stable set in $G$ and thus $T \cap V_1 = \emptyset$.
\end{proof}

Any connected component $C$ of $G[V_1 \cup T]$ is a \emph{weak component of $G$} if every vertex belongs to $V_2$ or $T$ (in other words, if no vertex of $C$ belongs to $V_1$ and not to $V_2$). The only apparent weak components on Figure~\ref{fig:V1} are encircled. The \emph{size} of a component of $G[V_1 \cup T]$ is the number of vertices of $V_1$ it contains. Let $C_w$ be the set of weak components of $G$ of size less than $\frac{1}{\epsilon}$. Let $S_w$ be the set of vertices of $V_2$ that belong to an element of $C_w$. Let $U$ be the set of vertices of degree at least $k-M$ with a neighbor in $S_w$.

%and since that neighbor has degree at least $k-M$ thus belongs to $U$,

We first need the following two results (Theorem~\ref{th:bkw97} corresponds to Theorem 3 of \cite{bkw97}).
\begin{theorem}\cite{bkw97}\label{th:bkw97}
For any bipartite multigraph $G$, if $L$ is a color assignment such that $\forall (u,v)\in E, |L(u,v)|\geq max(d(u),d(v))$, then $G$ is $L$-edge-choosable.
\end{theorem}
\begin{lemma}\label{lem:alphadegenerate}
For any bipartite multigraph $H$ with vertex set $A,B \neq \emptyset$. For $\alpha >0$, if for every subset $B' \subseteq B$ and $A'=N(B') \subseteq A$, there exists a vertex $u \in A'$ with $d_{B'}(u)<\alpha$, then $\alpha |A|>|B|$.
\end{lemma}
\begin{proof}
By induction on $|B|$. If $|B|<\alpha$, since $|A|\geq 1$, the conclusion holds. If $|B|\geq \alpha$, there exists $u \in A$ with $d(u)<\alpha$. We apply the induction hypothesis to the graph $H \setminus (\{u\}\cup N(u))$. It follows that $\alpha (|A|-1)>|B|-\alpha$, hence the result. 
\end{proof}

\begin{lemma}\label{lem:Cm}
The graph $G$ satisfies $|C_w|\leq \frac{1}{\epsilon} \times |\{v \in V \ | \ d_G(v) \geq k-M\}|$.
\end{lemma}

\begin{proof}
Assume by contradiction that $|C_w|> \frac{1}{\epsilon}\times |\{v \in V \ | \ d(v) \geq k-M\}|$.

Recall that by Lemma~\ref{lem:V1}, every vertex of $S_w \subseteq V_1$ has a unique neighbor in $U$.
Let $D$ be the bipartite multigraph whose vertex set is $V(D)=U \cup C_w$, and whose edge set is in bijection with $S_w$: for every element $v \in S_w$, we add an edge $(u,w)$, where $u$ is the element of $U$ adjacent to $v$ and $w$ is the element of $C_w$ to which $v$ belongs.
%Informally, we consider the bipartite multigraph obtained by contracting every element of $C_w$ into a single vertex, while maintaining all the cut edges.

For $A=\{v \in V \ | \ d(v) \geq k-M\}$, $B=C_w$ and $\alpha=\frac{1}{\epsilon}$, we have $|B|>\alpha|A|$. So by Lemma~\ref{lem:alphadegenerate}, there is a subset $C'_w$ of $C_w$ such that the subgraph $D'$ induced in $D$ by $C_w'$ and its neighbors $U'$ satisfies $\forall u \in U'$, $d_{D'}(u) \geq \frac{1}{\epsilon}$.

Let $S'_w$ (resp. $T'$) be the set of vertices of $S_w$ (resp. $T$) that belong to an element of $C'_w$.

We color by minimality $G \setminus (S'_w\cup T')$. Note that every vertex $v$ of $S_w'$, belonging to $V_2$, is adjacent to exactly one vertex $u$ of degree at least $k-M$, and that all its other neighbors $v_1,\ldots,v_{d(u)-1}$ are vertices of $T$ whose other neighbors $w_1,\ldots,w_{d(u)-1}$ are in $V_1$. Since the element $C$ of $C'_w \subseteq C_w$ to which $v$ belongs is a connected component of $G[V_1 \cup T]$, all the $v_i$'s and $w_i$'s belong to $C \in C'_w$. Consequently, for every $i$, we have $v_i \in T'$ and $w_i \in S'_w$. Thus $v$ has at most $k+1-d_{D'}(u)$ constraints, hence $v$ has at least $d_{D'}(u)$ colors available. To color the vertices of $S'_w$, it is sufficient to list-color the edges of $D'$, where every edge is assigned the same list of colors as the vertex of $S'_w$ it is in bijection with.

By assumption, every element of $C'_w$ contains at most $\frac{1}{\epsilon}$ vertices of $V_2$, so it has degree at most $\frac{1}{\epsilon}$ in $D$ thus in $D'$. Moreover, every vertex of $U'$ has degree at least $\frac{1}{\epsilon}$ in $D'$. Thus for every edge $(u,v)$ of $D'$, with $u \in U$ and $v \in C'_w$, we have $max(d_{D'}(u),d_{D'}(v))=d_{D'}(u)$. So $D'$ is a bipartite multigraph whose every edge has a list assignment of size at least $max(d_{D'}(u),d_{D'}(v))$. We apply Theorem~\ref{th:bkw97} to color the vertices of $S'_w$. It then remains to color the vertices of $T'$. These are vertices of degree $2$ whose both neighbors are in $S'_w$. But all the vertices of $S'_w$ are of degree at most $M$. So the vertices of $T'$ have at most $2 \times M \leq k$ constraints, and we can color the vertices of $T'$, a contradiction.
\end{proof}

\subsection{Discharging rules}\label{sect:dis}

Let $R_1$, $R_2$, $R_3$ and $R_g$ ('$g$' stands for 'global') be four discharging rules (see Figure~\ref{fig:rules}): for any vertex $x$,

\begin{itemize}
\item Rule $R_1$ is when $d(x)=2$ and its two neighbors $a$ and $b$ are such that $d(a)=2$ and $d(b) \geq M$, and the other neighbor $c$ of $a$ is not in $V_1$. Then $x$ gives $\frac{\epsilon}{2}$ to $a$.
\item Rule $R_2$ is when $3 \leq d(x) \leq M-1$ and $x \not\in V_1$. If $x$ has a neighbor $a$ of degree $2$ whose other neighbor is $y$,
\begin{itemize}
\item Rule $R_{2.1}$ is when $d(y)=2$. Then $x$ gives $1-\frac{3 \epsilon}{2}$ to $a$.
\item Rule $R_{2.2}$ is when $3 \leq d(y) < M$. If $y \not\in V_1$, then $x$ gives $\frac{1-\epsilon}{2}$ to $a$. If $y \in V_1$, then $x$ gives $1-\epsilon$ to $a$. 
\end{itemize}
\item Rule $R_3$ is when $M \leq d(x)$. Then $x$ gives $1-\frac{\epsilon}{2}$ to each of its neighbors.
\item Rule $R_g$ states that every vertex of degree at least $k-M$ gives an additional $\frac{1}{\epsilon}$ to an initially empty common pot, and every weak component of $G$ of size less than $\frac{1}{\epsilon}$ receives $1$ from this pot.
\end{itemize}

\captionsetup[subfloat]{labelformat=empty}
\begin{figure}[!h]
\centering
\subfloat[][$R_1$]{
\centering
\begin{tikzpicture}[scale=1.2]
\tikzstyle{whitenode}=[draw,circle,fill=white,minimum size=8pt,inner sep=0pt]
\tikzstyle{blacknode}=[draw,circle,fill=black,minimum size=6pt,inner sep=0pt]
\tikzstyle{tnode}=[draw,ellipse,fill=white,minimum size=8pt,inner sep=0pt]
\tikzstyle{texte} =[fill=white, text=black]
\draw (0,0) node[tnode] (b) [label=90:$b$] {\small{$M^+$}}
-- ++(-90:1cm) node[blacknode] (x) [label=right:$x$] {}
-- ++(-90:1cm) node[blacknode] (a) [label=right:$a$] {}
-- ++(-90:1cm) node[whitenode] (y) [label=-90:$c$][label=right:$\neg V_1$] {};

\draw (x) edge [post,bend right]  node [label=left:$\frac{\epsilon}{2}$] {} (a);
\end{tikzpicture}
%\caption{$R_1$}
\label{fig:r1}
}
\qquad
\subfloat[][$R_{2.1}$]{
\centering
\begin{tikzpicture}[scale=1.1]
\tikzstyle{whitenode}=[draw,circle,fill=white,minimum size=8pt,inner sep=0pt]
\tikzstyle{blacknode}=[draw,circle,fill=black,minimum size=6pt,inner sep=0pt]
\tikzstyle{tnode}=[draw,ellipse,fill=white,minimum size=8pt,inner sep=0pt]
\tikzstyle{texte} =[fill=white, text=black] 

\draw (0,0) node[tnode] (x) [label=90:$x$][label=right:$\neg V_1$] {\small{
\Large\textcolor{white}{$\frac{\textcolor{black}{3^+}}{\textcolor{black}{(M-1)^-}}$}}}
%$3^+$\vspace{2cm}$(M-1)^-$}}
-- ++(-90:1cm) node[blacknode] (a) [label=right:$a$] {}
-- ++(-90:1cm) node[blacknode] (y) [label=right:$y$] {}
-- ++(-90:1cm) node[whitenode] (z) {};

\draw (x) edge [post,bend right]  node [label=left:$1-\frac{3\epsilon}{2}$] {} (a);
\end{tikzpicture}
%\caption{$R_2$}
\label{fig:r2}
}
\subfloat[][$R_{2.2}$]{
\centering
\begin{tikzpicture}[scale=1.1]
\tikzstyle{whitenode}=[draw,circle,fill=white,minimum size=8pt,inner sep=0pt]
\tikzstyle{blacknode}=[draw,circle,fill=black,minimum size=6pt,inner sep=0pt]
\tikzstyle{tnode}=[draw,ellipse,fill=white,minimum size=8pt,inner sep=0pt]
\tikzstyle{texte} =[fill=white, text=black]
\draw (0,0) node[tnode] (x) [label=90:$x$][label=right:$\neg V_1$] {\small{
\Large\textcolor{white}{$\frac{\textcolor{black}{3^+}}{\textcolor{black}{(M-1)^-}}$}}}
%$3^+$ $(M-1)^-$}}
-- ++(-90:1cm) node[blacknode] (a) [label=right:$a$] {}
-- ++(-90:1cm) node[tnode] (y) [label=-90:$y$][label=right:$\neg V_1$] {\small{\Large\textcolor{white}{$\frac{\textcolor{black}{3^+}}{\textcolor{black}{(M-1)^-}}$}}}
%$3^+$ $(M-1)^-$}}
;

\draw (x) edge [post,bend right]  node [label=left:$\frac{1-\epsilon}{2}$] {} (a);

\draw (3.6,0) node[tnode] (x) [label=90:$x$][label=right:$\neg V_1$] {\small{
\Large\textcolor{white}{$\frac{\textcolor{black}{3^+}}{\textcolor{black}{(M-1)^-}}$}}}
%$3^+$ $(M-1)^-$}}
-- ++(-90:1cm) node[blacknode] (a) [label=right:$a$] {}
-- ++(-90:1cm) node[tnode] (y) [label=-90:$y$][label=right:$V_1$] {\small{
\Large\textcolor{white}{$\frac{\textcolor{black}{3^+}}{\textcolor{black}{(M-1)^-}}$}}}
%$3^+$ $(M-1)^-$}}
;

\draw (x) edge [post,bend right]  node [label=left:$1-\epsilon$] {} (a);
\end{tikzpicture}
\label{fig:r3}
}
\qquad
\subfloat[][$R_3$]{
\centering
\begin{tikzpicture}[scale=1.15]
\tikzstyle{whitenode}=[draw,circle,fill=white,minimum size=8pt,inner sep=0pt]
\tikzstyle{blacknode}=[draw,circle,fill=black,minimum size=6pt,inner sep=0pt]
\tikzstyle{tnode}=[draw,ellipse,fill=white,minimum size=8pt,inner sep=0pt]
\tikzstyle{texte} =[fill=white, text=black]

\draw (0,0) node[tnode] (x) [label=90:$x$] {\small{$M^+$}}
-- ++(-90:1cm) node[whitenode] (a) {};

\draw[white] (0,-2) node[tnode] (y) [label=90:$x$] {};

\draw (x) edge [post,bend right]  node [label=left:$1-\frac{\epsilon}{2}$] {} (a);
\end{tikzpicture}
\label{fig:r4}
}
\caption{Discharging rules $R_1$, $R_2$, and $R_3$ for Theorem~\ref{thm:m145}.}
\label{fig:rules}
\end{figure}
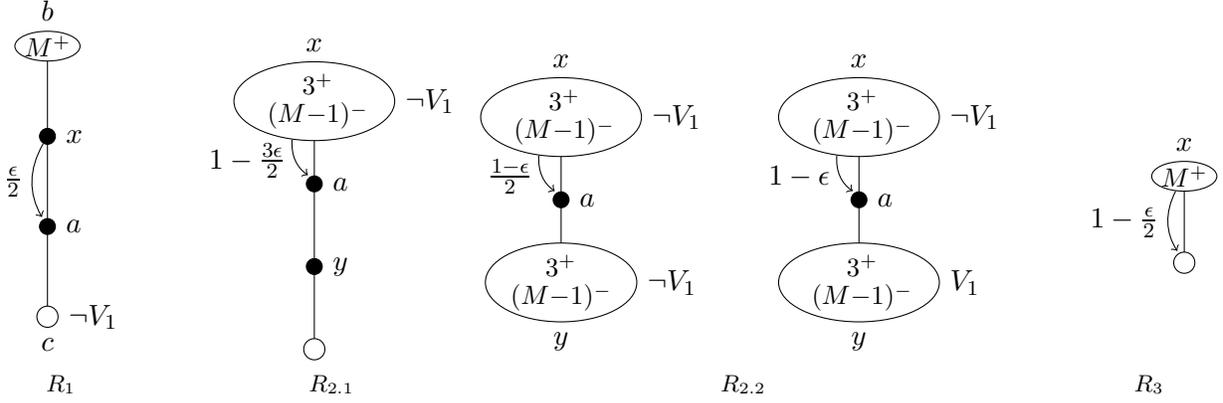
\captionsetup[subfloat]{labelformat=parens}

We use these discharging rules to prove the following lemma:

\begin{lemma}\label{lem:rules}
Graph $G$ satisfies $\mad(G) \geq 3-\epsilon$.
\end{lemma}

\begin{proof}
We attribute to each vertex $v$ a weight equal to $d(v)-3+\epsilon$, and apply discharging rules $R_1$, $R_2$, $R_3$ and $R_g$. We show that all the vertices of $G \setminus (T \cup V_1)$ have a non-negative weight in the end, and that each connected component of $G[T \cup V_1]$ has a non-negative total weight.

By Lemma~\ref{lem:Cm}, the common pot has a non-negative value, and Rule $R_g$ is valid.

Let $x$ be a vertex of $G \setminus (T \cup V_1)$. By Configuration \textbf{($C_1$)}, we have $d(x) \geq 2$.
\begin{enumerate}
\item $d(x)=2$.\\Vertex $x$ has an initial weight of $-1+\epsilon$. We prove that it receives at least $1-\epsilon$. Let $u_1$ and $u_2$ be its two neighbors. We consider two cases depending on whether one of them is of degree at least $M$.
\begin{enumerate}
\item \textit{$d(u_1) \geq M$ or $d(u_2) \geq M$}.\\Consider w.l.o.g. that $d(u_1)\geq M$. By $R_3$, vertex $u_1$ gives $1-\frac{\epsilon}{2}$ to $x$. Vertex $x$ gives at most $\frac{\epsilon}{2}$ to $u_2$ by $R_1$. So $x$ receives at least $1-\epsilon$.
\item \textit{$d(u_1)< M$ and $d(u_2) < M$}.\\Assume that $u_1$ or $u_2$ is of degree $2$. Consider w.l.o.g. that $d(u_1)=2$. Then $u_1$ belongs to $V_1$ by definition, and the other neighbor of $u_1$ is of degree at least $M$. Since $u_1 \in V_1$ and $x \not\in T$, then $u_2 \not\in V_1$ and we have $M \geq d(u_2) \geq 3$. By $R_{1}$ and $R_{2.1}$, vertex $u_1$ gives $\frac{\epsilon}{2}$ to $x$, and $u_2$ gives $1-\frac{3\epsilon}{2}$. So $x$ receives $1-\epsilon$ and gives no weight away.
If both $u_1$ and $u_2$ have degree at least three, then since $x \not\in T$, at most one of $u_1$ and $u_2$ is in $V_1$ and $R_{2.2}$ applies. So vertices $u_1$ and $u_2$ give a total of $1-\epsilon$ to $x$, and $x$ gives no weight away.
\end{enumerate}
\item $3 \leq d(x)\leq M-1$.\\Vertex $x$ has an initial weight of $d(x)-3+\epsilon \geq \epsilon$. Let $u_1, \ldots, u_q$ denote its neighbors of degree $2$ whose other neighbor is of degree at most $M-1$, where $u_1, \ldots, u_p$ denote its neighbors of degree $2$ whose other neighbor belongs to $V_1$ (note that $p$ may be equal to $0$ when $x$ has no such neighbor, and that $q$ may be equal to $p$). We consider two cases depending on $q$.
\begin{enumerate}
\item $q \leq d(x)-3$.\\Then $x$ gives at most $(d(x)-3)\times(1-\epsilon) \leq d(x)-3+\epsilon$ by $R_2$.
\item $q \geq d(x)-2$.\\Then, by Configuration \textbf{($C_3$)}, vertex $x$ has a neighbor $v$ with $d(v) \geq \frac{k-M+2}{2}\geq M$ (recall that $k \geq 3\times M$). By Rule $R_3$, vertex $x$ receives $1-\frac{\epsilon}{2}$ from $v$. We consider two cases depending on $p$.
\begin{enumerate}
\item $p \leq d(x)-3$. By Rule $R_2$, $x$ gives at most $(d(x)-3)\times(1-\epsilon)+ 2 \times \frac{1-\epsilon}{2} \leq d(x)-3+\epsilon+(1-\frac{\epsilon}{2})$.
\item $p \geq d(x)-2$. Since $x \not\in V_1$, we have $p=q=d(x)-2$. By Rule $R_2$, $x$ gives at most $(d(x)-2)\times(1-\epsilon) \leq d(x)-3+\epsilon+(1-\frac{\epsilon}{2})$.
\end{enumerate}
\end{enumerate}
\item $M \leq d(x) \leq k-M-1$.\\By Rule $R_3$, vertex $x$ gives at most $d(x)\times(1-\frac{\epsilon}{2})$. Since $M=\frac{6}{\epsilon}$, we have $d(x) \times \frac{\epsilon}{2} \geq 3 \geq 3 - \epsilon$, so $x$ has a non-negative final weight.
\item $k-M \leq d(x)$.\\By Rules $R_3$ and $R_g$, vertex $x$ gives at most $\frac{1}{\epsilon}+d(x)\times(1-\frac{\epsilon}{2})$. Since $k\geq \frac{3}{\epsilon^2}$, $M=\frac{6}{\epsilon}$ and $\epsilon \leq \frac{1}{20}$, we have $d(x) \times \frac{\epsilon}{2} \geq (\frac{3}{\epsilon^2}-\frac{6}{\epsilon})\times \frac{\epsilon}{2} = \frac{3}{2\epsilon} - 3 \geq \frac{1}{\epsilon}+10-3 \geq \frac{1}{\epsilon}+3 - \epsilon$, so $x$ has a non-negative final weight.
\end{enumerate}

Therefore, every vertex of $G \setminus (T \cup V_1)$ has a non-negative final weight. It remains to consider vertices of $G[T \cup V_1]$. Let $C$ be a connected component of $G[T \cup V_1]$. Let $s$ be the size of $C$. Note that $s \geq 1$.

If $s=1$, then $C$ consists of a single vertex $u$ of degree $2$ in $G$ and that is adjacent to a vertex $v$ of degree at least $k-M$ and $1$-linked to a vertex $w$ of degree less than $M$. Thus, by $R_1$ and $R_3$, vertex $u$ has an initial weight of $-1+\epsilon$, receives $1-\frac{\epsilon}{2}$ from $v$, and gives $\frac{\epsilon}{2}$ to its neighbor of degree $2$: it has a final weight of $0$. We assume from now on that $s \geq 2$.

No vertex of $C$ gives weight. Indeed, only $R_1$ can apply on a vertex of $C$ since all the others apply on vertices of degree at least $M$ or on vertices of degree at least $3$ but not in $V_1$. 
If $R_1$ applies on a vertex $x$ of $C$, then $d(x)=2$ and its two neighbors $a$ and $b$ are such that $d(b)\geq M$ and $d(a)=2$, where the other neighbor $c$ of $a$ is not in $V_1$. Since $d(b) \geq M$, we have $b \not\in V_1$ and $x \not\in T$, so $x \in V_1$. Then, since $V_1$ is a stable set by Lemma~\ref{lem:V1}, we have $a \not\in V_1$, and $s=1$, a contradiction with our assumption.

Every vertex $u$ in $C \cap V_1$ has exactly one neighbor of degree at least $k-M$ (not in $C$) and receives $1-\frac{\epsilon}{2}$ from it. Thus the weight $W$ of $C$ (without taking $R_g$ into account) is as follows. We denote by $N(C)$ the set of vertices that do not belong to $C$ but are adjacent to a vertex in $C$.

\begin{align*}
W & \geq \sum_{v\in C}(d(v)-3+\epsilon)+s\times(1-\frac{\epsilon}{2})\\
& \geq \sum_{v\in T \cap C}(d(v)-3+\epsilon)+\sum_{v\in V_1 \cap C}(d(v)-3+\epsilon)+s\times(1-\frac{\epsilon}{2})\\
& \geq \sum_{v\in T \cap C}(-1+\epsilon)+\sum_{v\in V_1 \cap C}d_C(v)+\sum_{v\in V_1 \cap C}d_{N(C)}(v)+\sum_{v\in V_1 \cap C}(-3+\epsilon)+s\times(1-\frac{\epsilon}{2})\\
& \geq \sum_{v\in T \cap C}(-1+\epsilon)+\sum_{v\in V_1 \cap C}d_C(v)+(s+|N(C) \setminus U|)+s\times(-3+\epsilon)+s\times(1-\frac{\epsilon}{2})
\end{align*}
Remember that the vertex set of $C$ is the union of $V_1 \cap C$ and $T \cap C$, which are stable sets. Also, the two neighbors of a vertex in $T$ belong to $C$, so $\sum_{v\in V_1 \cap C}d_C(v)=\sum_{v\in T \cap C}d_C(v)=2|T\cap C|$. Since $C$ is a connected component, we have $|T\cap C| \geq |V\cap C|-1=s-1$. Then,
\begin{align*}
W & \geq|T\cap C|\times(-1+\epsilon)+2|T\cap C|+|N(C) \setminus U|+s\times(-1+\frac{\epsilon}{2})\\
& \geq(-1-\epsilon+\frac{3 \epsilon s}{2})+|N(C) \setminus U|
\end{align*}

We consider three cases depending on whether $C$ is weak and $s < \frac{1}{\epsilon}$.
\begin{enumerate}
\item \emph{$C$ is a weak component of $G$ and $s < \frac{1}{\epsilon}$}.\\By $R_g$, component $C$ receives an extra weight of $1$. Thus, it has a final weight of $1+W\geq 1+(-1-\epsilon+\frac{3 \epsilon s}{2})= -\epsilon+\frac{3 \epsilon s}{2} > 0$.
\item \emph{$C$ is a weak component of $G$ and $s \geq \frac{1}{\epsilon}$}.\\Then $C$ has a final weight of $W \geq -1-\epsilon+\frac{3 \epsilon s}{2}\geq -1 -\epsilon +\frac{3 \epsilon}{2 \times \epsilon}\geq 0$.
\item \emph{$C$ is not a weak component of $G$}.\\There is at least a vertex $v$ in $(V_1 \cap C) \setminus V_2$, so $v$ has a neighbor $x$ of degree $2$ whose other neighbor is not in $V_1$. Then $x$ is not in $T$, thus not in $C$. So $x\in N(C)\setminus U$. Then the final weight of $C$ is $W\geq (-1-\epsilon+\frac{3 \epsilon s}{2})+1\geq 0$.
\end{enumerate}

Consequently, after application of the discharging rules, every vertex $v$ of $G\setminus\{V_1 \cup T\}$ has a non-negative final weight, and every connected component $C$ of $G[V_1 \cup T]$ has a non-negative final total weight, meaning that $\sum_{v \in G} (d(v)-3+\epsilon) \geq 0$. Therefore, $\mad(G) \geq 3- \epsilon$. This completes the proof of Lemma~\ref{lem:rules}, and thus of Theorem~\ref{thm:m145}.

\end{proof}

\section{List injective coloring}\label{sect:inj}

A \emph{list injective k-coloring} of a graph is a (not necessarily proper) list k-coloring of its vertices such that two vertices with a common neighbor are of different color, or, in other words, such that no vertex has two neighbors with the same color.
Note that the proof for Theorem~\ref{thm:m145} also work, with close to no alteration, for list injective coloring with one color less. Indeed, the discharging part does not depend on the problem considered, and the configuration part can easily be checked to work also for this as, though one less color is available, every critical vertex has at least one less constraint since already colored neighbors do not count anymore. There is no reason to think that this would be the case for any discharging proof about list coloring of the square, but it happens to be the case most often. 

We thus obtain the following theorem.

\begin{theorem}\label{thm:inj}
There exists a function $f$ such that for any $\epsilon>0$, every graph $G$ with $mad(G)<3-\epsilon$ and $\Delta(G) \geq f(\epsilon)$ satisfies $\chi^i_\ell(G)= \Delta(G)$.
\end{theorem}

Theorem~\ref{thm:inj} is optimal in the same sense as Theorem~\ref{thm:m145} by the graph family described in Figure~\ref{fig:mad3}.

\section{Conclusion}\label{sect:concl}

For any $C \geq 1$, we asked for the supremum $M(C)$ such that any graph $G$ with $\mad(G) < M(C)$ and sufficiently large $\Delta(G)$ (depending only on $\mad(G)$) satisfies $\chi^2_\ell(G) \leq \D+C$. It was already known~\cite{blp13} that $lim_{C \rightarrow \infty}M(C)=4$. We proved here that $M(1)=3$, and conjectured that $M(C)=\frac{4C+2}{C+1}$.

It might be a good approach to try, for every fixed $C$, to adapt the same proof outline. However, rather than prove incremental results, it would be more interesting to look for a general proof that would work for every $C$, or maybe only for every large enough $C$.

The proof that $lim_{C \rightarrow \infty}M(C)=4$~\cite{blp13} is quite short and simple. However, finding the exact value of $M(C)$ might still be difficult for large $C$. Indeed, the proof does not even involve lower bounds on $\Delta(G)$.

This leads to a similar question (with no constraint on $\Delta(G)$).
\begin{question}
What is, for any $C \geq 1$, the maximum $m(C)$ such that any graph $G$ with $\mad(G) < m(C)$ satisfies $\chi^2_\ell(G) \leq \D+C$?
\end{question}

Obviously, we have $m(C)\leq M(C)<4$ for every $C$. Our proof that $lim_{C \rightarrow \infty}M(C)=4$ was actually a proof that $lim_{C \rightarrow \infty}m(C)=4$. Note that $m(1)=m(2)=2$, as a cycle $C_5$ of length five has $\mad(C_5)=2$, $\Delta(C_5)=2$ and $\chi^2(C_5)=5$, while any graph $G$ with $\mad(G)<2$ is a forest and thus satisfies $\chi^2_\ell(G)=\Delta(G)+1$. So there is a significant gap for $C=1,2$. What about general $C$? Does there exist some $C$ for which $m(C)=M(C)$?

\bibliographystyle{plain}

\end{document}